\documentclass[11pt]{article}
\usepackage[utf8]{inputenc}
\usepackage[T1]{fontenc}
\usepackage{lmodern,microtype}
\usepackage[letterpaper,margin=1in]{geometry}
\usepackage{amsmath,amssymb,amsfonts,amsthm,mathtools}
\usepackage{booktabs,array,enumitem,float,graphicx}
\graphicspath{{figures/}}
\usepackage[hidelinks]{hyperref}
\usepackage{url}
\renewenvironment{abstract}
 {\begin{center}\bfseries Abstract\end{center}\begin{quote}\normalsize}
 {\end{quote}}
\newtheorem{theorem}{Theorem}
\newtheorem{lemma}[theorem]{Lemma}
\newtheorem{thm}[theorem]{Theorem}
\newtheorem{lem}[theorem]{Lemma}
\newtheorem{cor}[theorem]{Corollary}
\theoremstyle{definition}
\newtheorem{Defi}[theorem]{Definition}
\theoremstyle{remark}

\floatstyle{ruled}
\newfloat{algorithm}{tbp}{loa}
\floatname{algorithm}{Algorithm}

\newcommand{\whereisliteral}{\textsc{WhereIsTheLiteral}}
\newcommand{\FL}{\textsc{FindLiteral}}
\newcommand{\IL}{\textsc{IsLiteral}}
\newcommand{\FRB}{\textsc{FindRelevantBlock}}
\newcommand{\BBS}{\textsc{BlockBinarySearch}}
\newcommand{\UJ}{\textsc{UniformJunta}}
\newcommand{\diff}{\operatorname{diff}}
\newcommand{\Bad}{\mathsf{Bad}}
\newcommand{\D}{\mathcal D}
\newcommand{\U}{\mathcal U}
\title{Near-Linear Regimes for One-Sided\\Distribution-Free Junta Testing}
\author{Xiaojin Zhang}
\hypersetup{pdfauthor={Xiaojin Zhang}}
\date{}
\begin{document}
\pagenumbering{roman}
\begin{titlepage}
\maketitle
\thispagestyle{empty}

\begin{abstract}
We study adaptive one-sided distribution-free testing of Boolean
$k$-juntas, with membership-query access to the function and sample
access to an arbitrary unknown distribution. For general Boolean
functions, we give a tester using
$O(k^2\log(k+1)+k/\epsilon)$ queries, independently of the ambient
dimension. The tester maintains disjoint relevant blocks certified by
observed disagreements. Within each fresh batch, it jointly decodes
literal signatures for a small family of random base masks, then uses
linear combinations to generate $\Theta(1/\epsilon)$ probes. Their ideal
counterparts are pairwise independent, yielding constant conditional
progress probability; a global $O(k)$ batch budget gives the additive
query bound. Consequently, the complexity is
$O((k/\epsilon)\log(k+1))$ when $\epsilon\le1/k$, and $O(k/\epsilon)$
when $\epsilon\le1/(k\log(k+1))$. We also analyze an indexed-table implementation that reuses previously
computed block orientations. For functions symmetric on an unknown
set of relevant variables, it uses
$O(k\log^2(k+1)/\epsilon)$ queries: safely checked blocks contain a
single globally relevant variable, making their cached nonliteral
sides reusable. For arbitrary functions, a stable-deletion analysis
establishes soundness of the same implementation at
$O(k^2\log(k+1)/\epsilon)$ queries. All variants always accept
$k$-juntas, and every rejection is certified by $k+1$ disjoint relevant
blocks. The fresh-batch construction and the reuse analysis use the
same certified-block framework, with different randomness and
stopping rules.
\end{abstract}
\vfill
\end{titlepage}
\pagenumbering{arabic}

\section{Introduction}
A Boolean function $f:\{0,1\}^n\to\{0,1\}$ is a $k$-junta if it
depends on at most $k$ coordinates. Distribution-free junta testing
asks whether $f$ is a $k$-junta or is $\epsilon$-far from every
$k$-junta under an arbitrary unknown distribution $\D$. The tester
receives samples from $\D$ and may query $f$ at inputs of its choice.
We study adaptive testers with one-sided error: every $k$-junta must
be accepted. This requirement makes an observed collection of
$k+1$ disjoint relevant blocks a useful rejection certificate. Our
focus is the cost of finding such certificates while repeatedly
localizing representative literals inside the current blocks.

Property testing studies whether an object has a specified property or
is far from every object having it, using a limited number of
observations~\cite{goldreich1998property}. Distribution-free testing
measures distance under an unknown distribution accessible through
sampling~\cite{goldreich1998property,halevy2007distribution}.
Early work on testing Boolean literals and juntas established query
bounds independent of the ambient
dimension~\cite{parnas2002testing,fischer2004testing}. Junta testing
also serves as a component of testers for concise representations,
including decision trees, formulas, and
circuits~\cite{diakonikolas2007concise}.

In the uniform model, Blais~\cite{blais2008improved}
obtained a nonadaptive tester using
$\widetilde O(k^{3/2})/\epsilon$ queries.
He subsequently gave a one-sided adaptive tester using
$O(k/\epsilon+k\log(k+1))$ queries~\cite{blais2009testing}.
At constant accuracy, Sa\u{g}lam's adaptive lower bound
$\Omega(k\log k)$~\cite{sauglam2018near}
matches its dependence on $k$. For arbitrary $\D$, Liu, Chen, Servedio, Sheng, and Xie~\cite{liu2019distribution} gave a
one-sided adaptive tester with $O(k^2\log(k+1)+k^2/\epsilon)$ queries,
independent of $n$, and an exponential nonadaptive lower bound.
Bshouty~\cite{bshouty2019almost} obtained a two-sided adaptive bound of
$O((k/\epsilon)\log(k/\epsilon))$. These results motivate separating
the cost of structural information about relevant blocks from the
number of distributional probes needed to detect distance.

The relevant-block framework permits two forms of reuse. Fresh batches
reuse random bases within a fixed structural state; indexed tables
retain orientations across structural updates. The former yields an
additive accuracy term for general functions. The latter requires
care with adaptive dependence: a success bound for a fixed state
does not automatically remain valid after the reused rows have
influenced the choice of that state.

\subsection{Our Main Results}
All logarithms are natural unless a base is displayed. We assume
$1\le k<n$ and $0<\epsilon\le1$; when $k\ge n$, the tester accepts
without queries. Query complexity counts membership queries and samples
from $\D$, including evaluations at sampled points. It does not include
the time needed to construct length-$n$ query strings.

The principal general-input guarantee is the fresh-batch bound
$O(k^2\log(k+1)+k/\epsilon)$ in Theorem~\ref{thm:batch-main}.
It follows from joint literal decoding, pairwise independent ideal
probes, and a global progress budget. Theorem~\ref{thm:symmetric-main}
analyzes indexed reuse under hidden symmetry, while
Theorem~\ref{thm:general-main} gives a larger-budget guarantee for
the same indexed implementation on arbitrary functions. Together,
these results distinguish the cost of generating new probes from
the conditions under which orientations can be retained across
adaptive updates.

Compared with the bound $O(k^2\log(k+1)+k^2/\epsilon)$ obtained in
the analysis of the certified-block tester in
\cite[Theorem 4.3]{liu2019distribution}, the fresh-batch guarantee
retains a quadratic structural term and reduces the accuracy-dependent
term to $k/\epsilon$. This separation is useful when accuracy dominates
the cost: for $\epsilon\le1/(k\log(k+1))$, the total is $O(k/\epsilon)$.

\begin{thm}[Additive-accuracy variant]\label{thm:batch-main}
A fresh-batch variant, specified in Section~\ref{batch:section},
is a one-sided distribution-free tester for arbitrary Boolean functions
with query complexity
\begin{equation}\label{eq:additive-main}
O\!\left(k^2\log(k+1)+\frac{k}{\epsilon}\right).
\end{equation}
It changes the mask generation and the stopping rule; it does not use
the original table across different batches.
\end{thm}

\begin{thm}[Near-linear indexed reuse under hidden symmetry]
\label{thm:symmetric-main}
Suppose $f(x)=\phi(\sum_{i\in S}x_i)$ for an unknown set
$S\subseteq[n]$ and a Boolean map $\phi$.
The indexed-table tester in Section~\ref{sec:tester}, with
\begin{equation}\label{eq:short-budget}
T_{\rm sym}=\left\lceil\frac{2(\ln k+6)}{\epsilon}\right\rceil,
\end{equation}
always accepts every $k$-junta, including those outside this function
class. If $f$ satisfies the displayed symmetry condition and is
$\epsilon$-far from every $k$-junta under an arbitrary $\D$, it rejects
with probability at least $2/3$. Its worst-case query complexity is
\[
O\!\left(\frac{k\log^2(k+1)}{\epsilon}\right).
\]
The algorithm does not need to know $S$ or $\phi$.
\end{thm}

The symmetry condition includes parity, conjunctions, disjunctions,
majority, and unweighted thresholds on hidden variables. It is an
assumption for soundness of the short-budget mode, not a property
verified by the algorithm.

\begin{thm}[General indexed-table guarantee]\label{thm:general-main}
For arbitrary Boolean $f$, the same tester with
\begin{equation}\label{eq:general-budget}
T_{\rm gen}=\left\lceil\frac8\epsilon(k\ln2+\ln6)\right\rceil
\end{equation}
is a one-sided distribution-free $k$-junta tester. Its worst-case
query complexity is $O(k^2\log(k+1)/\epsilon)$.
\end{thm}

\begin{cor}[Near-linear accuracy regimes]\label{cor:accuracy-main}
For arbitrary $f$, the fresh-batch variant has
\[
\epsilon\le\frac1k
\quad\Longrightarrow\quad
Q=O\!\left(\frac{k}{\epsilon}\log(k+1)\right),
\]
and
\[
\epsilon\le\frac1{k\log(k+1)}
\quad\Longrightarrow\quad Q=O(k/\epsilon).
\]
These are bounds relative to the scale $k/\epsilon$; $\epsilon$ may
depend on $k$.
\end{cor}

The general fixed-accuracy lower bound $\Omega(k\log k)$ does not,
by itself, prove a matching lower bound in the small-$\epsilon$ regimes
or for the hidden-symmetric promise class. We therefore do not claim
global optimality of these results.

\begin{table}[tb]
\centering
\caption{Established bounds and the guarantees proved here. Accuracy
is constant in the lower-bound row. All rows except that lower bound
are one-sided upper bounds unless marked two-sided.}
\label{tab:results}
\begin{tabular}{@{}p{.17\textwidth}p{.25\textwidth}p{.46\textwidth}@{}}
\toprule
Distribution & Scope / source & Query bound \\
\midrule
Uniform & Blais~\cite{blais2009testing} & $O(k/\epsilon+k\log(k+1))$ \\
Uniform & Sa\u{g}lam~\cite{sauglam2018near} & $\Omega(k\log k)$, any adaptive tester \\
Arbitrary & Liu, Chen, Servedio, Sheng, and Xie~\cite{liu2019distribution} & $O(k^2\log(k+1)+k^2/\epsilon)$ \\
Arbitrary & Bshouty~\cite{bshouty2019almost} & $O((k/\epsilon)\log(k/\epsilon))$, two-sided \\
Arbitrary & Indexed reuse, hidden symmetry & $O(k\log^2(k+1)/\epsilon)$ \\
Arbitrary & Indexed reuse, general $f$ & $O(k^2\log(k+1)/\epsilon)$ \\
Arbitrary & Fresh batches, general $f$ & $O(k^2\log(k+1)+k/\epsilon)$ \\
Arbitrary & Fresh batches, $\epsilon\le1/k$ & $O((k/\epsilon)\log(k+1))$ \\
Arbitrary & Fresh batches, $\epsilon\le1/(k\log(k+1))$ & $O(k/\epsilon)$ \\
\bottomrule
\end{tabular}
\end{table}

\subsection{Further Related Work}
Bshouty's framework for testing concise representations gives
two-sided distribution-free testers for juntas and several related
classes, with a $\widetilde O(k/\epsilon)$ query bound for
juntas~\cite{bshouty2020concise}. In that work, symmetry of a class
means closure under permutations of the variables. This differs from
the function-level hidden-symmetry promise in
Theorem~\ref{thm:symmetric-main}. The testers studied here preserve
one-sided error through actual relevance certificates.

Recent work considers other access and distance models.
Chen, Pires, Pitassi, and Servedio~\cite{chen2025relative} study
relative-error testing with uniform samples from satisfying
assignments. Beretta, Harms, and Koch~\cite{beretta2025feature} obtain
optimal sample complexity for junta testing and feature selection in
a sample-only model. These models differ from the membership-query
and arbitrary-distribution sample access used here, so their bounds
are not directly comparable with our query guarantees.

Bshouty and Haddad~\cite{bshouty2026small} give a
uniform-distance junta tester with $O(k2^k+1/\epsilon)$ queries and
use pairwise-independent assignments to reduce simulation costs.
The fresh-batch construction in Section~\ref{batch:section} also
uses pairwise independence. Its analysis combines joint signature
decoding with explicit block certificates to handle distance under
an arbitrary unknown distribution and retain one-sided error.
We do not claim that pairwise independence itself is a new testing
tool.

\subsection{Technical Overview}
All variants maintain disjoint relevant blocks with actual oracle
disagreements. A structural check either splits a block into two
certified children or declares its anchored restriction special.
Specialness means uniform $1/8$-closeness to a nonconstant literal;
erroneous declarations are accounted for separately. Every update
increases the block count, so $k+1$ blocks certify rejection and
ensure one-sidedness independently of decoding errors.

Section~\ref{batch:section} gives the general-input fresh-batch tester.
At a fixed safe state, it draws $O(\log(2/\epsilon))$ random base
masks per block and jointly decodes the representative literal's
membership signature. Distinct nonzero binary linear combinations
then produce $\Theta(1/\epsilon)$ probes. In the ideal experiment,
these masks are pairwise independent fair subsets outside the
representative coordinates; a second-moment bound supplies a
discrepancy with constant probability. An unconditional coupling
accounts for decoding errors without conditioning away mask
independence. Block binary search and fresh complement tests turn
the discrepancy into certified progress. Fresh bases in each batch
give a conditional progress bound at every adaptive state, and a
global $O(k)$ batch budget yields the additive query bound.

More quantitatively, let $L=\lceil16/\epsilon\rceil$. The ideal
probe indicators have common mean $p\ge\epsilon$ and are pairwise
independent, so the probability of no discrepancy is at most
$1/(L\epsilon)\le1/16$. The joint decoder's total error over the
current blocks is at most $1/16$, uniformly for every fixed family of
bases. An inside hit requires a second certificate; two new complement
tests miss it with probability at most $1/16$. Coupling the actual and
ideal experiments without conditioning on decoding success therefore
gives progress probability at least $13/16$. With the virtual progress
indicators defined in Section~\ref{batch:global}, a budget of $M$
batches gives
\[
 \Pr\!\left[\sum_{t=1}^{M}V_t\le k\right]
 \le 2^k(19/32)^M.
\]
Structural-check errors are budgeted separately. This conditional
argument permits $M=O(k)$ at constant confidence, rather than requiring
every search stage to succeed with high probability.

The indexed-table tester in Section~\ref{sec:tester} instead retains
orientations for unchanged block instances. Under hidden symmetry,
a safely checked block has exactly one globally relevant coordinate;
its cached nonliteral sides are therefore globally irrelevant.
For general functions, the analysis independently retains half the
row indices and replays the execution. Retaining all selected
successful indices reproduces the terminal state, while omitted
rows remain independent at that replay state. This stable-deletion
argument gives the larger general budget rather than assuming fresh
conditional behavior of reused rows.

\paragraph{Organization.}
Section~\ref{sec:preliminaries} fixes the query model and certificate
invariants. Sections~\ref{sec:tester}--\ref{sec:analysis} present the
indexed-table tester and its two soundness analyses.
Section~\ref{batch:section} gives the complete general-input fresh-batch
algorithm, its correctness proof, and its query bound.
Appendix~\ref{sec:literal-check} supplies the structural literal-check
procedure used by the indexed implementation.

\section{Preliminaries}\label{sec:preliminaries}
For a set $B\subseteq[n]$, write $\bar B=[n]\setminus B$ and
$x^{(B)}$ for $x$ with exactly the coordinates of $B$ flipped.
For two strings, put $\diff(x,y)=\{i:x_i\ne y_i\}$.
A distinguishing pair for $B$ consists of $x,y$ with
$f(x)\ne f(y)$ and $\diff(x,y)\subseteq B$.
The algorithms below store the stronger \emph{full-flip certificate}
$a,a^{(B)}$ with $f(a)\ne f(a^{(B)})$.
Every certified block contains a globally relevant coordinate.
The global relevant-coordinate set is
\[
\operatorname{supp}(f)=\{i:\exists x, f(x)\ne f(x^{(\{i\})})\}.
\]

Distance is $\operatorname{dist}_{\D}(f,h)=
\Pr_{x\sim\D}[f(x)\ne h(x)]$, and
$\operatorname{dist}_{\D}(f,\mathcal J_k)$ is its minimum over
all $k$-juntas $h$. A function is $\epsilon$-far if this minimum is
at least $\epsilon$. The global property of being a $k$-junta does
not depend on $\D$; its distance from that class does.

\begin{Defi}[One-sided distribution-free tester]
A tester always accepts every $k$-junta. On every input covered by its
soundness guarantee with $\operatorname{dist}_{\D}(f,\mathcal J_k)
\ge\epsilon$, it rejects with probability at least $2/3$.
For the general modes, soundness covers every Boolean function; for
the short-budget mode, it covers the hidden-symmetric promise class.
\end{Defi}

A block instance consists of $B$, a fixed certificate anchor $a_B$,
the restriction
\[
g_B(u)=f(u_B(a_B)_{\bar B}),\qquad u\in\{0,1\}^B,
\]
and a fresh identifier assigned deterministically in creation order.
Changing $B$ or its anchor creates a new instance. An instance is
\emph{special} if $g_B$ is uniformly $1/8$-close to a nonconstant
literal $u_i$ or $1-u_i$. A \emph{safe search state} has at most
$k$ checked instances, each actually special. Passing a randomized
check is an algorithmic status and is not an exact proof of specialness.

We use standard Fourier notation for Boolean
functions~\cite{odonnell2014analysis} in the hidden-symmetric
analysis. The concentration bounds for literal localization are
instances of Hoeffding's bounded-variable
inequality~\cite{hoeffding1963probability}; all required estimates
are derived below.

\section{A tester for distribution-free junta}\label{sec:tester}
\subsection{Problem Statement}
The tester has membership-query access to $f$ and independent sample
access to $\D$. Uniform strings and fair subsets can be generated
using internal randomness. Each theorem specifies the function class
or parameter range for its soundness guarantee.

Distance can change substantially with the distribution. The upper
row of Figure~\ref{fig_distance} gives a distribution-dependent distance example:
\[
f(x_1,x_2,x_3)=
\begin{cases}1,&x_1=0,\\x_2,&x_1=1.\end{cases}
\]
Under the uniform distribution, its distance from every $1$-junta
is at least $1/4$, with equality for $x_2$ and $1-x_1$.
If $\D$ is uniform on $x_1=0$, then $f=1$ almost surely under
$\D$, so the distance is zero. This is a statement about distance,
not about a change in the function's global support.

\begin{figure}[tbp]
\centering
\includegraphics[width=.9\textwidth]{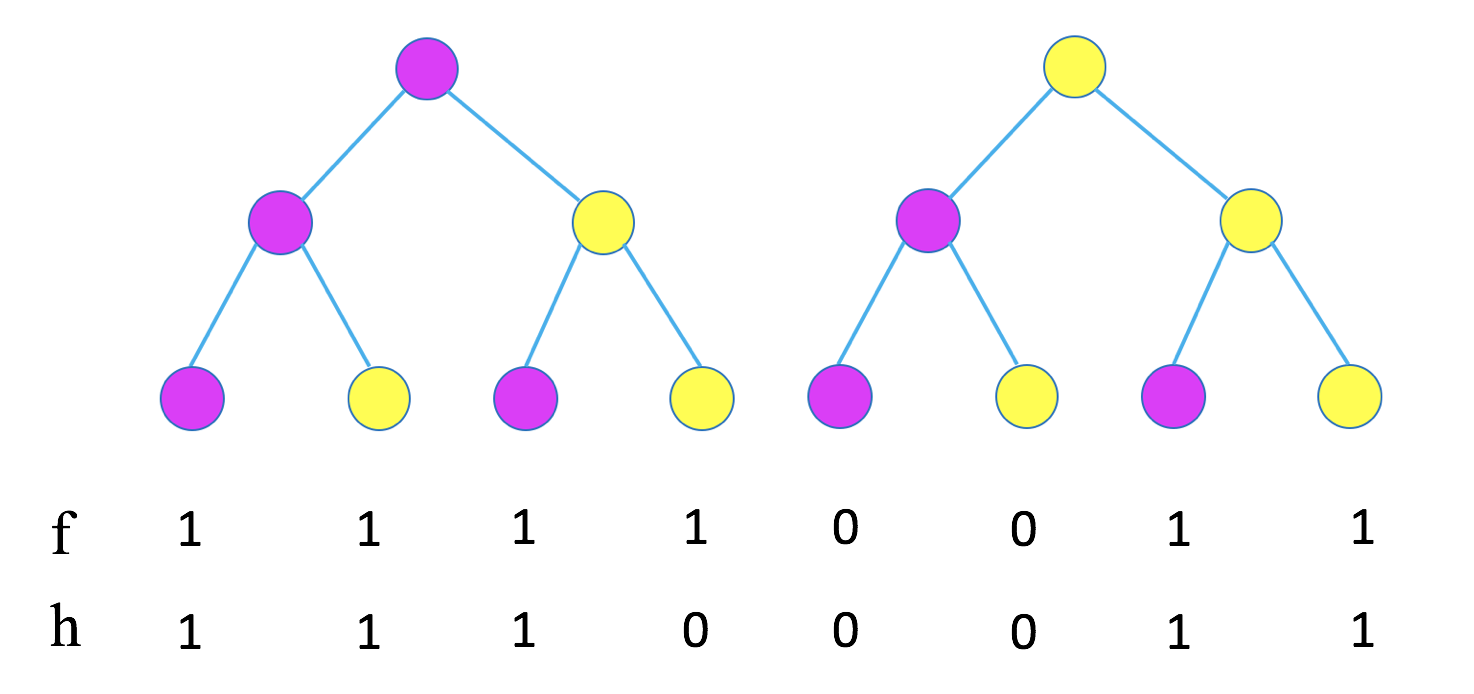}
\caption{Distance and variation under different distributions. Pink nodes
represent bit $0$ and yellow nodes bit $1$; the eight branches are
$000,001,\ldots,111$ from left to right. The upper row gives $f$:
its uniform distance from the $1$-juntas is $1/4$, whereas it agrees
with the constant $1$ on the first four branches supporting $\D$.
The lower row gives the auxiliary function $h$ used below to illustrate
failure of subadditivity under a different distribution $\D_0$.}
\label{fig_distance}
\end{figure}

\subsection{Overview of Previous Approaches}
Random-partition approaches under product distributions use a
structural guarantee that fixing a few parts leaves substantial
variation~\cite{blais2009testing}. In the distribution-free setting,
fixing a few relevant coordinates is sufficient for the partial
randomization guarantee below, but fixing entire large relevant blocks
need not be sufficient. This motivates locating a representative
literal within each checked block.

The standard uniform influence arguments cannot be transferred
without checking their hypotheses. The lower row of
Figure~\ref{fig_distance} defines
\[
h(x_1,x_2,x_3)=
\begin{cases}1-x_2x_3,&x_1=0,\\x_2,&x_1=1.\end{cases}
\]
Let $\D_0$ be concentrated at $000$, and define
$\mathbb V_{\D_0,\U}(A)=
\Pr_{x\sim\D_0,w\sim\U}[h(x)\ne h(x_{\bar A}w_A)]$.
Then
\[
\mathbb V_{\D_0,\U}(\{2\})=
\mathbb V_{\D_0,\U}(\{3\})=0,
\qquad \mathbb V_{\D_0,\U}(\{2,3\})=1/4.
\]
Thus the corresponding subadditivity assertion is false for arbitrary
$\D$. Our proofs use certificates and fresh conditional progress,
or stable deletion of unused rows, instead.

\begin{lem}[Partial randomization; cf.\ {\cite[Lemma 3.2]{liu2019distribution}}]
\label{lem_liu}
If $f$ is $\epsilon$-far from every $k$-junta under $\D$ and
$|I|\le k$, then independent $x\sim\D$ and $w\sim\U$ satisfy
\begin{equation}\label{eq_leml_1}
\Pr[f(x)\ne f(x_Iw_{\bar I})]\ge\epsilon.
\end{equation}
Equivalently, for an independent fair subset $R\subseteq\bar I$,
\begin{equation}\label{eq_leml_2}
\Pr[f(x)\ne f(x^{(R)})]\ge\epsilon.
\end{equation}
In particular, the weaker bound $\epsilon/2$ also holds.
\end{lem}
\begin{proof}
For each fixed $w$, $h_w(x)=f(x_Iw_{\bar I})$ is a $k$-junta,
so its disagreement with $f$ under $\D$ is at least $\epsilon$.
Average over uniform $w$. Conditional on $x$, independently fair
flips on $\bar I$ have the same law as uniform resampling there.
\end{proof}

\subsection{The Algorithm}
The stage framework is retained, with its state operations made
explicit. Let $\mathcal B$ contain unchecked instances,
$\mathcal K$ checked instances, and $\mathcal A\subseteq\mathcal K$
checked instances whose tables have not yet been built.
The map $\mathcal G$ stores each instance's fixed certificate and
anchor. Write $w=|\mathcal B|+|\mathcal K|$.
Every structural update increases $w$ by one.
Figure~\ref{fig: flowchart} summarizes this framework; the algorithms
below specify the state updates and indexed-cache operations.

\begin{figure}[H]
\centering
\includegraphics[width=.83\textwidth]{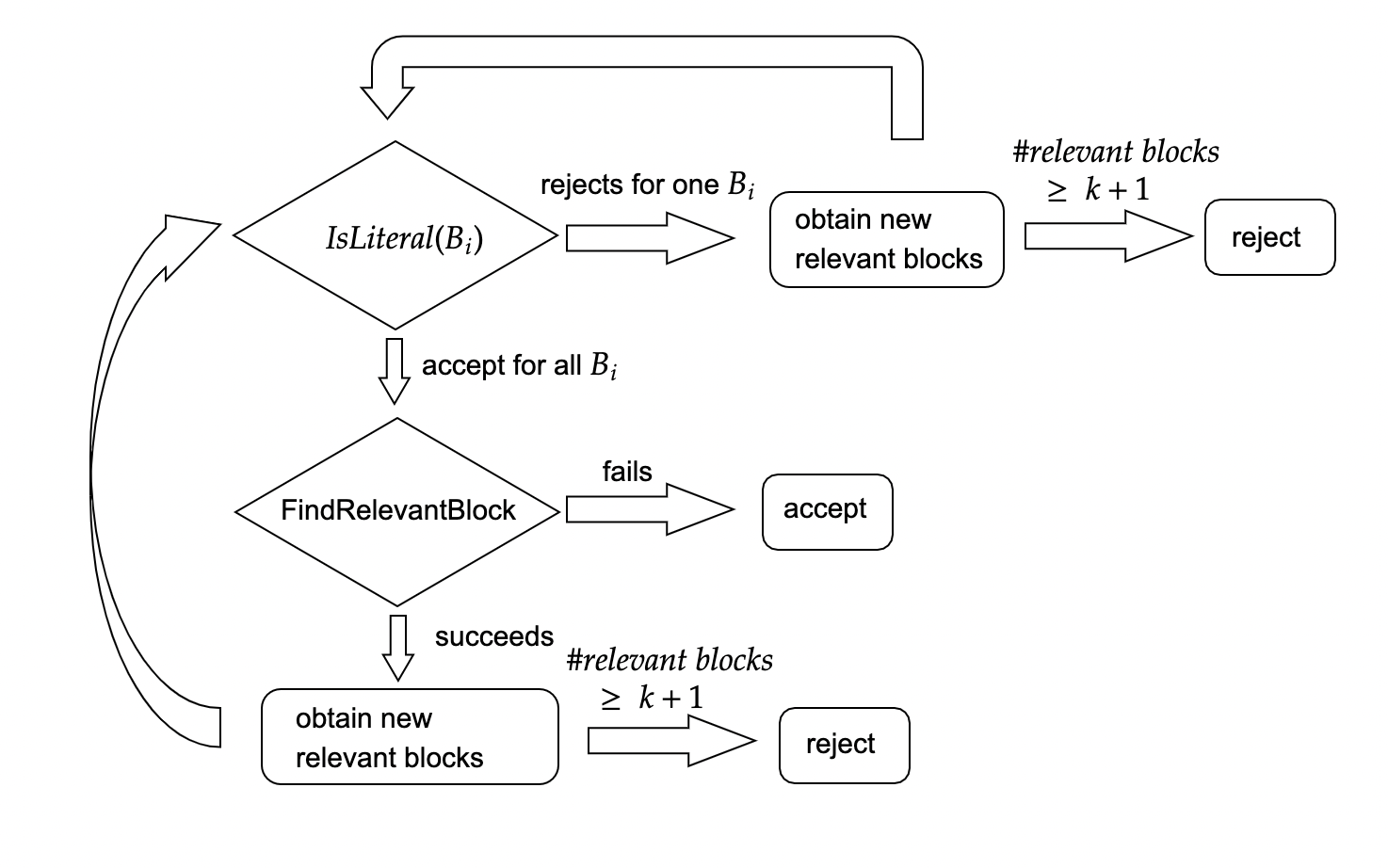}
\caption{The relevant-block framework. A certified rejection
of $\IL$ replaces one pending block by two disjoint relevant subblocks,
which need not cover the parent. Once all pending instances pass,
$\FRB$ either returns a certified addition or split, or the tester
accepts. The main algorithm specifies instance bookkeeping and the
budget $T$. This flowchart applies to the indexed-table tester in the
main text; Section~\ref{batch:section} specifies the global stopping
rule of the fresh-batch variant.}
\label{fig: flowchart}
\end{figure}

\begin{algorithm}[tb]
\caption{Distribution-Free Junta Testing with budget $T$}
\label{Framework: Main Algorithm}

\begin{enumerate}[leftmargin=*,itemsep=3pt]
\item Initialize $\mathcal A=\mathcal B=\mathcal K=\varnothing$,
      the indexed cache $\mathcal N=\varnothing$, and $w=0$.
\item If $w=k+1$, return reject.
\item If an unchecked instance $B\in\mathcal B$ exists, call
      $\IL(B,a_B)$. On pass, remove $B$ from $\mathcal B$ and add
      it to $\mathcal K$ and $\mathcal A$.
      On a certified split, remove the parent and install the two
      children as unchecked instances with their separate anchors.
      Recompute $w$ and return to step 2.
\item Otherwise call $\FRB(f,T,\mathcal A,\mathcal K,\mathcal G,
      \mathcal N)$. If it returns false, return accept.
\item Apply its certified addition or split. New instances are
      unchecked; removed parents lose their cache entries. Recompute
      $w$ and return to step 2.
\end{enumerate}
\end{algorithm}

An instance is checked only once. When a pending block splits, already
passed instances remain in $\mathcal A$ until their tables are built;
that obligation is not discarded by the split. Parent and child
instances always have distinct identifiers.

The split interface is
\begin{equation}\label{eq:split-interface}
\varnothing\ne Z_1,Z_2\subseteq B,\quad Z_1\cap Z_2=\varnothing,
\quad f(z_i)\ne f(z_i^{(Z_i)})\quad(i=1,2).
\end{equation}
The children need not cover $B$, and their anchors need not be equal.
Unused coordinates return to the outside region. This weaker interface
is sufficient for counting disjoint relevant blocks and is what the
certificate-producing uniform tester guarantees after shrinking its
distinguishing pairs to their difference sets.

\paragraph{Cache convention.}
For each passing instance and row $t\in[T]$, store
$\mathcal N[B,t]=(L_B^t,N_B^t,s_B^t)$, where
$f(s_B^t)\ne f((s_B^t)^{(L_B^t)})$, or store $\bot$ on failure.
All $T$ entries of a new instance are generated independently across
rows, before probing the search. A failed entry does not stop the
generation of other entries. The row witness does not overwrite the
instance's anchor. Unselected rows cause no structural update.
Old instances reuse the same entries with the same row labels.

\begin{algorithm}[tb]
\caption{$\FRB(f,T,\mathcal A,\mathcal K,\mathcal G,\mathcal N)$}
\label{Framework: SubAlgorithmFindRelevantBlock}

\begin{enumerate}[leftmargin=*,itemsep=3pt]
\item For each $B\in\mathcal A$ and $t=1,\ldots,T$, draw a
      fair partition of $B$ and call $\FL$ on the frozen $g_B$.
      Store its certified $(L_B^t,N_B^t,s_B^t)$ or $\bot$.
      Continue filling every table even after a failure.
\item Set $\mathcal A=\varnothing$. Scan $t=1,\ldots,T$ in order.
      If some current entry $\mathcal N[B,t]=\bot$, skip this row.
\item For a usable row, draw a fresh $x\sim\D$ and a fresh fair
      subset $C\subseteq[n]\setminus\bigcup_{B\in\mathcal K}B$.
      Form
      \[
      y=x^{(C\,\cup\,\bigcup_{B\in\mathcal K}N_B^t)}.
      \]
\item If $f(x)\ne f(y)$, run $\BBS$ on the nonempty chunks among
      $C$ and $\{N_B^t:B\in\mathcal K\}$, together with the saved
      $L_B^t$ certificates, and return its certified update.
\item Return false if all $T$ rows fail or are skipped.
\end{enumerate}
\end{algorithm}

Each search uses new outside masks and distribution samples, even when
it reads an old table row. Initially $\mathcal K$ is empty, so a row
simply flips a fair subset of $[n]$.
Figure~\ref{fig_constructing_y} illustrates the mask construction.

\begin{figure}[tbp]
\centering
\includegraphics[width=.45\textwidth]{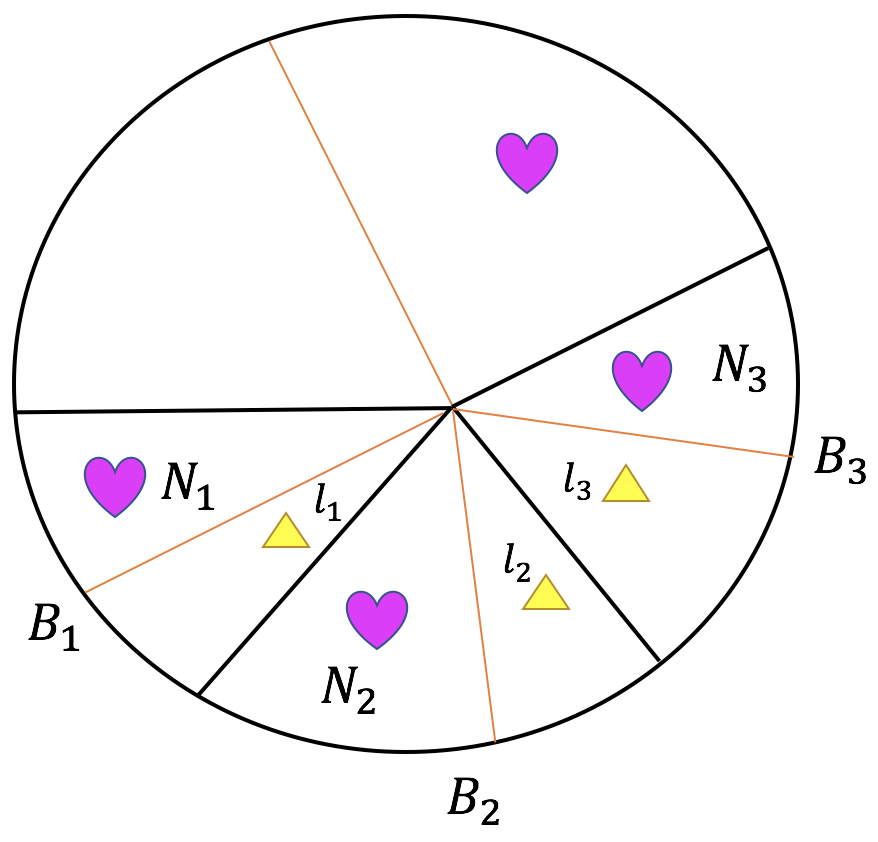}
\caption{Construction of $y$ in a usable row with correctly oriented
representative literals. The drawing labels the three checked
blocks $B_1,B_2,B_3$; triangles denote their frozen restrictions'
literals. The heart-marked regions are flipped: $N_i^t$ inside each
block and a fresh fair subset $C$ of the outside region. Other
coordinates are unchanged. For general functions a checked block may
contain additional globally relevant coordinates; under hidden symmetry
and safe checks, every $N_i^t$ is globally irrelevant.}
\label{fig_constructing_y}
\end{figure}

\begin{algorithm}[tb]
\caption{$\BBS$: block-level binary search and certified update}
\label{Framework: SubAlgorithm1}

\begin{enumerate}[leftmargin=*,itemsep=3pt]
\item The input endpoints $a,b$ have opposite function values and
      differ on the union of a list of disjoint nonempty chunks.
\item While more than one chunk remains, partition the list into two
      nearly equal sublists. Query $c$, obtained from $a$ by flipping
      the union of the first sublist. If $f(c)\ne f(a)$, keep
      $(a,c)$ and that sublist; otherwise keep $(c,b)$ and the second
      sublist. Retain the queried endpoint values.
\item At the remaining chunk $H$, obtain $z,z^{(H)}$ with opposite
      values. If $H=C$, return addition of the certified block $C$.
\item Otherwise $H=N_B^t$ for a current parent $B$. Return removal
      of $B$ and insertion of $L_B^t,N_B^t$, using respectively the
      saved witness $s_B^t$ and the binary-search witness $z$.
      Both children are unchecked instances.
\end{enumerate}
\end{algorithm}

The maintained endpoints differ exactly on the remaining union.
Thus a leaf supplies a full-flip certificate. At most
$\lceil\log_2(k+1)\rceil$ additional queries are made. In the
inside-hit case, the parent is removed from $\mathcal K$ as well as
from any auxiliary state; it is not retained alongside its children.

Figure~\ref{operation reuse} illustrates the inside-hit update of
the general indexed-table tester. In the safe hidden-symmetric regime,
every cached nonliteral side is globally irrelevant, so a successful
block search instead adds a block from the outside region.

\begin{figure}[tbp]
\centering
\begin{minipage}[t]{.46\textwidth}
\centering
\includegraphics[width=.94\linewidth]{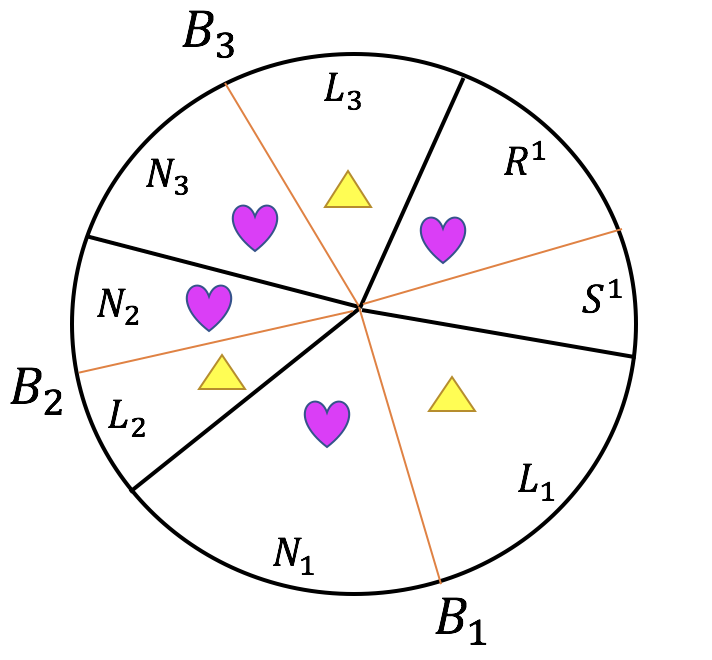}
\par\smallskip (a) Three checked blocks.
\end{minipage}\hfill
\begin{minipage}[t]{.46\textwidth}
\centering
\includegraphics[width=.94\linewidth]{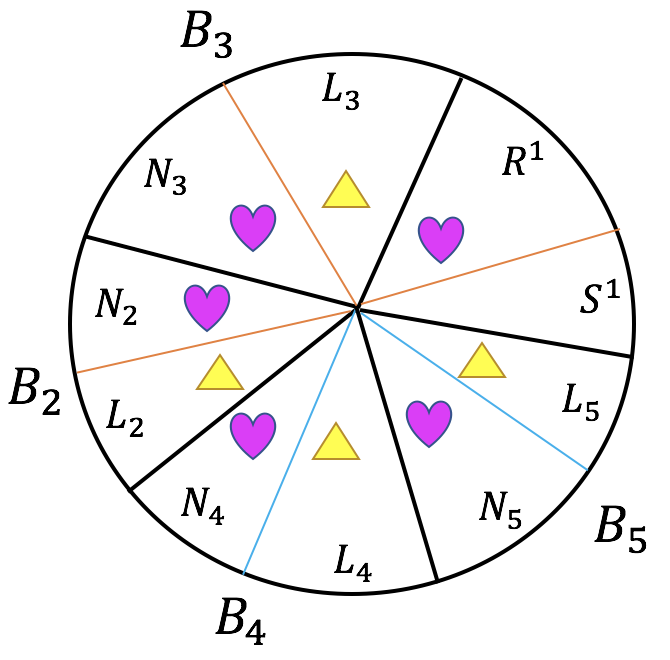}
\par\smallskip (b) Four blocks after the split and checks.
\end{minipage}
\caption{Indexed orientation reuse in the general-function framework.
(a) A selected row gives a disagreement whose block binary search
isolates $N_1^t$. Remove the parent $B_1$ and install the certified
children $B_4=N_1^t$ and $B_5=L_1^t$ as unchecked instances.
(b) If both children subsequently pass $\IL$, build new tables for
them. Unchanged frozen instances $B_2,B_3$ retain their entries at the
same row labels. In the drawing $R^1$ denotes an outside flip mask
and $S^1$ its complement; outside masks and distribution samples are
fresh in every search. The outside regions in the two panels are
schematic: the repeated labels $R^1,S^1$ do not denote identical
random subsets. The inside-hit transition cannot occur for
hidden-symmetric inputs while all literal-check acceptances are safe.}
\label{operation reuse}
\end{figure}

\section{Subroutines for finding relevant blocks}\label{sec:subroutines}
\subsection{Subroutine for testing the literal}
We use the certificate-producing uniform tester stated in
\cite[Theorem 2.3]{liu2019distribution}, based on
Blais~\cite{blais2009testing}.

\begin{lem}[Certified uniform testing]\label{lem_Uniform_Testing}
$\UJ(g,j,\gamma)$ always accepts a $j$-junta and rejects a function
$\gamma$-far from every $j$-junta with probability at least $2/3$.
It uses $O(j/\gamma+j\log(j+1))$ queries. Each rejection supplies
$j+1$ pairwise disjoint blocks with actual distinguishing pairs.
For $j=1$ and $\gamma=1/8$, its cost is an absolute constant.
\end{lem}

If a returned pair $(u_i,v_i)$ differs only within a block $D_i$,
set $Z_i=\diff(u_i,v_i)\subseteq D_i$. Lifting $u_i,v_i$ with
the frozen outside anchor produces precisely a full-flip certificate
for $Z_i$. This uses no extra queries.

\begin{lem}[Literal check]\label{Literal_Lem}
The procedure $\IL$ in Appendix~\ref{sec:literal-check} always
returns valid certificates when it splits. If the input restriction
has distance greater than $1/8$ from every literal, its probability
of passing is at most $1/(28k)$. Its query cost is
$O(\log(k+1))$, and it always passes an exact literal.
\end{lem}

\subsection{Subroutine for locating the literal}\label{SecLiteralOrientate}
\subsubsection{Comparison with the previous work}
The subroutine $\whereisliteral$ in
\cite{liu2019distribution} uses closeness $1/(8k)$ to obtain
an orientation error of order $1/k$ with constant query cost.
Here the two sides have a constant separation in their expected
flip-disagreement scores when $\gamma=1/8$. Amplifying that
comparison costs $O(\log(k+1))$ queries. The closeness parameter
remains constant in both the indexed-table and fresh-batch variants.

\subsubsection{FindLiteral Algorithm}
\begin{algorithm}[tb]
\caption{$\FL(g_B,P,Q)$ for a fixed partition $B=P\sqcup Q$}
\label{Framework:Algorithm for LO}

\begin{enumerate}[leftmargin=*,itemsep=3pt]
\item Set $q=\lceil128\ln(64k)\rceil$ and both scores to zero.
\item For each $a=1,\ldots,q$ and each side $H\in\{P,Q\}$,
      draw a fresh uniform $u\in\{0,1\}^B$ and query
      $g_B(u),g_B(u^{(H)})$. Increase the score of $H$ on a
      disagreement and save one observed disagreement pair for it.
\item If one score is strictly larger, return that side $L$ and its
      saved witness, lifted with the frozen outside anchor.
      The score is then positive, so such a witness has been observed.
\item On a tie, return fail. Set $N=B\setminus L$ on success.
\end{enumerate}
\end{algorithm}

\begin{lem}[Location for every fixed partition]\label{constant_gamma}
If $g_B$ is $1/8$-close to a literal on coordinate $i$, then for
every fixed partition $B=P\sqcup Q$, $\FL$ returns the side
containing $i$, with a true certificate, with probability at least
$1-1/(12k)$. Every nonfailure output has a true certificate, whether
its orientation is correct or not. For an exact literal the output
is deterministically correct.
\end{lem}
\begin{proof}
Both a uniform $u$ and its fixed-side flip have uniform marginals.
The literal changes on its containing side and is unchanged on the
other. A union bound on the two approximation errors shows that the
disagreement probabilities are respectively at least $3/4$ and at
most $1/4$. In round $a$, let $X_a$ be the literal-side indicator
minus the other-side indicator. Then $X_a\in[-1,1]$,
$\mathbb E X_a\ge1/2$, and the rounds are independent. Hoeffding's
inequality gives
\[
\Pr\left[\sum_{a=1}^qX_a\le0\right]
\le e^{-q/8}\le\frac1{12k}.
\]
A strict winning score is positive and supplies an already queried
certificate. For an exact literal, the two scores are exactly $q$
and zero. The argument includes partitions with an empty side.
\end{proof}

\section{Analysis of the Algorithm}\label{sec:analysis}
\subsection{Correctness Analysis of the Algorithm}
Every structural update adds a disjoint certified block or replaces
one certified block by two disjoint certified subblocks. Hence the
current blocks remain disjoint and genuinely relevant, even if a
statistical check or orientation is incorrect.

\begin{lem}[Unconditional one-sidedness]\label{correct2}
For every $T$, Algorithm~\ref{Framework: Main Algorithm} always
accepts every $k$-junta.
\end{lem}
\begin{proof}
Each current block contains a relevant coordinate, and the blocks are
pairwise disjoint. A $k$-junta cannot produce $k+1$ such blocks.
There are at most $k+1$ structural increases and at most $2k+1$
created instances before termination. All checks and searches have
finite budgets. Thus the algorithm terminates and can only accept.
\end{proof}

Let $\Bad$ be the event that an actual literal check passes a
restriction farther than $1/8$ from every literal. An instance is
checked only once, using fresh randomness after its anchor is fixed.
Conditioning on the preceding history of each actual check gives
\begin{equation}\label{eq:check-errors}
\Pr(\Bad)\le\frac{2k+1}{28k}\le\frac3{28}.
\end{equation}
This is an adaptive union bound on the performed checks. Subsequent
proofs do not condition fresh randomness on the future event
$\Bad^c$.

\subsubsection{Near-linear testing for hidden symmetric functions}
\label{sec:hidden-symmetric}

The short search budget also admits a complete analysis for a structural
promise on the function, without any assumption on the distribution.
We call a Boolean function \emph{hidden symmetric} if there are an unknown
set $S\subseteq[n]$ and a map
$\phi:\{0,\ldots,|S|\}\to\{0,1\}$ such that
\begin{equation}
 f(x)=\phi\left(\sum_{i\in S}x_i\right).
 \label{eq:hidden-symmetric-definition}
\end{equation}
This family includes parity, AND, OR, majority, and unweighted threshold
functions on an unknown subset of the coordinates. The tester does not
need to know $S$ or $\phi$. If $\phi$ is nonconstant, every coordinate in
$S$ is globally relevant: there is an integer $j$ with
$\phi(j)\ne\phi(j+1)$, and setting the other coordinates in $S$ to have
weight $j$ gives a distinguishing pair for any prescribed coordinate of
$S$. Thus $S$ can be taken to be the global support of $f$. Constant
functions have empty support and are already $k$-juntas.

The following result uses the same block framework and indexed-table
reuse as the general tester. In particular, the disjoint
split certificates need not cover the parent block. Each block instance
has a frozen anchor and is checked once. Its orientation table is built
once, using independent random streams, and only the selected successful
search row changes the active blocks. The sampled point and the external
random mask in each search trial are fresh and independent of all cached
orientation information.

We now prove Theorem~\ref{thm:symmetric-main} with the short budget and literal-location sampling count specified above.

\begin{lemma}[Symmetric restrictions cannot approximate a literal]
\label{lem:symmetric-literal-gap}
Let $f$ satisfy \eqref{eq:hidden-symmetric-definition}. Fix a block $B$
and an external anchor $z_{\overline B}$, and write
$g_B(u)=f(u_Bz_{\overline B})$. If $r=|B\cap S|\ge2$, then the uniform
distance from $g_B$ to every literal on $B$ is at least
\[
 \frac{1-1/\sqrt2}{2}>\frac18.
\]
If $B$ has a full-flip certificate
$f(z)\ne f(z^{(B)})$ and $g_B$ is $1/8$-close to a literal, then
$|B\cap S|=1$ and $g_B$ is exactly a literal or its negation.
\end{lemma}

\begin{proof}
The anchored function $g_B$ depends only on $B\cap S$ and is symmetric
on these $r$ coordinates: it is obtained by adding the fixed integer
$\sum_{i\in S\setminus B}z_i$ to their Hamming weight before applying
$\phi$. Encode its values by $G_B(u)=(-1)^{g_B(u)}$, and use the
characters $\chi_i(u)=(-1)^{u_i}$. By symmetry, its singleton Fourier
coefficients on $B\cap S$ all have the same value $a$. Parseval's
identity gives
\[
 r a^2\le \sum_{A\subseteq B}\widehat G_B(A)^2=1.
\]
For $i\in B\cap S$, the minimum distance to the two literals on
coordinate $i$ is therefore
\[
 \min\{\Pr[g_B(u)\ne u_i],\Pr[g_B(u)\ne 1-u_i]\}
 =\frac{1-|a|}{2}
 \ge\frac{1-1/\sqrt r}{2}
 \ge\frac{1-1/\sqrt2}{2}.
\]
For $i\in B\setminus S$, the coefficient is zero because $g_B$ is
independent of $u_i$, so both literal distances are $1/2$. This proves
the first assertion. The certificate excludes $r=0$ and excludes a
constant restriction. The first assertion excludes $r\ge2$ whenever
$g_B$ is $1/8$-close to a literal. Hence $r=1$, and the nonconstant
Boolean function of the single remaining coordinate is an exact
literal or its negation.
\end{proof}

\begin{proof}[Proof of Theorem~\ref{thm:symmetric-main}]
Let $\mathsf{Bad}$ be the event that some \textbf{IsLiteral} call
accepts a frozen restriction whose distance from every literal exceeds
$1/8$. Every tested block has a genuine full-flip certificate. As in the
general analysis, there are at most $2k+1$ tested block instances before
termination. Conditional on the history preceding any call, its fresh
test randomness gives an erroneous acceptance probability at most
$1/(28k)$. Summing over these calls yields
\begin{equation}
 \Pr(\mathsf{Bad})\le\frac{2k+1}{28k}\le\frac3{28}.
 \label{eq:symmetric-bad}
\end{equation}

Consider a search call reached before the first erroneous acceptance,
and condition on the entire history at its start, including all cached
orientation tables of its current blocks. Every current block
$K\in\mathcal K$ passed \textbf{IsLiteral} without such an error.
Lemma~\ref{lem:symmetric-literal-gap} implies that $K$ contains exactly
one globally relevant coordinate $\ell_K$, and its frozen restriction
is an exact literal or its negation. For any partition of $K$, every
trial in \textbf{FindLiteral} has reward one on the side containing
$\ell_K$ and reward zero on the other side. Thus \textbf{FindLiteral}
always returns the correct side, never returns a tie, and supplies a
valid certificate. In every cached row,
\begin{equation}
 N_K^t\cap S=\varnothing,
 \qquad
 f(x^{(N_K^t)})=f(x)\quad\text{for every }x.
 \label{eq:symmetric-cache-irrelevant}
\end{equation}
This is a global irrelevance statement, rather than a statement only
about the frozen restriction. Consequently it remains valid regardless
of which past search rows selected the current blocks.

Set $I=\{\ell_K:K\in\mathcal K\}$. The current blocks are disjoint,
so $|I|=|\mathcal K|\le k$. The fresh mask $C_t$ independently flips
each coordinate outside their union with probability $1/2$. By
\eqref{eq:symmetric-cache-irrelevant}, the cached flips inside the
current blocks do not affect $f$. All relevant coordinates in
$S\setminus I$ lie outside those blocks and are independently
randomized by $C_t$. Thus, conditional on the history and cached
tables, the function values in this trial have the same joint
distribution as
\[
 \bigl(f(x),f(x_Iw_{\overline I})\bigr),
 \qquad x\sim\mathcal D,
 \quad w\sim\mathcal U,
\]
with fresh independent $x,w$. For every fixed $w$, the function
$x\mapsto f(x_Iw_{\overline I})$ is a $k$-junta. The distance
assumption therefore gives a discrepancy probability at least
$\epsilon$; in particular, the weaker bound $\epsilon/2$ suffices.
Whenever a discrepancy is found, block binary search supplies a real
certificate and increases the number of disjoint relevant blocks.

After any number of failed trials in this search call, the active
blocks have not changed, and the next sampled point and external mask
are again fresh. Iterating the conditional failure bound gives
\begin{equation}
 \Pr[\text{this search fails}\mid\text{safe history at its start}]
 \le(1-\epsilon/2)^T
 \le e^{-\epsilon T/2}
 \le\frac{e^{-6}}{k}
 \le\frac1{6(k+1)}.
 \label{eq:symmetric-stage-failure}
\end{equation}
There are at most $k+1$ search calls. A certified split or successful
search increases the number of disjoint relevant blocks by one, and
the tester rejects as soon as this number is $k+1$. Hence acceptance
before the first erroneous literal-test acceptance requires one of
these search calls to fail. Applying
\eqref{eq:symmetric-stage-failure} to the safe histories of their
starts and taking a union bound gives a probability at most $1/6$.
Together with \eqref{eq:symmetric-bad},
\[
 \Pr[\text{accept}]
 \le\frac3{28}+\frac16
 =\frac{23}{84}<\frac13.
\]
This argument conditions only on past histories; it does not condition
on the future event that no literal-test error ever occurs.

For a true $k$-junta, every certified relevant block contains a globally
relevant coordinate. Since the active blocks are disjoint and each
increase preserves genuine certificates, their number never exceeds
$k$, regardless of statistical errors. Rejection is therefore
impossible, establishing the one-sided guarantee.

Finally, each of the at most $2k+1$ block instances is checked once at
$O(\log(k+1))$ cost. Each passing instance builds its $T$-row
orientation table only once, with $O(\log(k+1))$ queries per row.
There are at most $(k+1)T$ fresh search trials and at most $k+1$
successful block binary searches, each costing $O(\log(k+1))$.
Counting distribution samples as queries, the total cost is
\[
 O\bigl(k\log(k+1)+kT\log(k+1)+kT+k\log(k+1)\bigr)
 =O\left(\frac{k\log^2(k+1)}{\epsilon}\right).
\]
\end{proof}

\paragraph{Scope of the guarantee.}
The theorem proves near-linear query complexity for the hidden
symmetric promise family under every distribution $\mathcal D$.
It neither requires a product distribution nor changes the cached
orientation tables. It does not assert this short-budget soundness
bound for general Boolean functions, and it makes no matching
lower-bound or global optimality claim for the promise family.

\subsubsection{General functions and the enlarged table budget}
\begin{lem}[An independent row at a fixed safe state]
\label{lem:fixed-row}
For an $\epsilon$-far function, a new independent row at a fixed
safe state produces a certified structural increase with probability
at least $\epsilon/4$.
\end{lem}
\begin{proof}
Let $I$ be the set of literal coordinates of the current restrictions.
For each random partition, define an ideal row using the actual
literal-containing sides. Its combined flip set $R^*$ is a fair
subset of $\bar I$. Indeed, within a block with literal coordinate
$i$, independently fair partition bits $P_j$ give flip indicators
$P_j\oplus P_i$ for $j\ne i$. The transformation
\[
(P_i,(P_j)_{j\ne i})\longmapsto
(P_i,(P_j\oplus P_i)_{j\ne i})
\]
is a bijection, so those indicators are jointly independent and fair.
The outside bits are independent as well.

Let $E$ be ideal disagreement and $A$ correct certified location in
every current block. Lemma~\ref{lem_liu} gives
$\Pr(E)\ge\epsilon/2$. For every fixed collection of partitions
$P$, sample $x$, and outside subset $C$, Lemma~\ref{constant_gamma}
and a union bound give
\[
\Pr(A\mid P,x,C)\ge1-\frac{|\mathcal K|}{12k}\ge\frac{11}{12}.
\]
Location samples are independent of $x,C$, so
\[
\Pr(A\cap E)
=\mathbb E_{P,x,C}[\mathbf1_E\Pr(A\mid P,x,C)]
\ge\frac{11}{12}\Pr(E)
\ge\frac{11\epsilon}{24}\ge\frac\epsilon4.
\]
On this event the actual row equals the ideal row, and binary search
returns a valid update. We never assume that conditioning on correct
location preserves the distribution of the partitions.
\end{proof}

\begin{lem}[Stable deletion]\label{lem:stable-deletion}
At the first accepting search of an execution, let $U$ be the set of
indices selected by preceding successful searches. Then $|U|\le k$.
For every $J\subseteq[T]$ containing $U$, replaying with only rows
in $J$ and the same literal-check streams reproduces the structural
history and the same terminal accepting state.
\end{lem}
\begin{proof}
Preindex independent random streams by instance identifier and row
for partitions and location, and by search-call number and row for
distribution samples and outside masks. Let $W_t$ contain all streams
with row label $t$, and $V$ the literal-check streams indexed by
instance identifier. Then $V,W_1,\ldots,W_T$ are independent.
Unused samples are merely deferred; this description causes no
extra oracle queries.

For fixed $J$, build only entries in $J$ and scan them in their
original order. Retain original row labels, deterministic instance
identifiers, and search-call numbers. Induct over structural updates.
Identical instances and anchors cause identical literal checks. At a
successful full search, its first successful index lies in $U\subseteq J$;
all earlier retained rows fail in both runs, so the same update is
selected and the same instances are created. Omitted rows cause no
update and do not change the retained entries. This remains true if
an index succeeds in more than one search. The replay therefore
reaches the same accepting state, where its retained rows all fail.
Each preceding search success increases the block count, giving
$|U|\le k$. The replay depends only on $V,W_J$.
\end{proof}

\begin{proof}[Soundness in Theorem~\ref{thm:general-main}]
Put $F=\{\mathsf{accept}\}\cap\Bad^c$. Independently retain each
row label with probability $1/2$, obtaining $J$. On $F$, $|U|\le k$,
so
\begin{equation}\label{eq:retained-successes}
\Pr(F\cap\{U\subseteq J\})
=\mathbb E[\mathbf1_F2^{-|U|}]
\ge2^{-k}\Pr(F).
\end{equation}
Fix $J$ and condition only on $V,W_J$, not on $F$ or on
$U\subseteq J$. The restricted replay is determined. If its terminal
state is not accepting and safe, it contributes no event to
$F\cap\{U\subseteq J\}$.
Otherwise its terminal state and search-call index are fixed.
Each held-out $W_t$, $t\notin J$, independently defines a complete
row at that state, with success probability at least $\epsilon/4$
by Lemma~\ref{lem:fixed-row}. Stable deletion implies that on
$F\cap\{U\subseteq J\}$ this is the actual terminal state, and
every held-out row fails there. Thus
\[
\Pr(F\cap\{U\subseteq J\}\mid J,V,W_J)
\le(1-\epsilon/4)^{T-|J|}.
\]
This is event containment, not conditioning on the full accepting
history. Averaging the independent retention coins gives
\[
\Pr(F\cap\{U\subseteq J\})
\le\mathbb E_J(1-\epsilon/4)^{T-|J|}
=(1-\epsilon/8)^T.
\]
Together with \eqref{eq:retained-successes},
\begin{equation}\label{eq:general-failure}
\Pr(F)\le2^k(1-\epsilon/8)^T
\le\exp(k\ln2-\epsilon T/8)\le\frac16
\end{equation}
for $T=T_{\rm gen}$. Adding \eqref{eq:check-errors} yields
$\Pr(\mathsf{accept})\le3/28+1/6=23/84<1/3$.
Random retention is only an analysis device; the actual algorithm
does not delete rows. This proof does not establish that the original
short budget fails on general functions.
\end{proof}

\subsection{Query Complexity Analysis of the Algorithm}
\begin{thm}[Cost of the indexed-table framework]\label{thm:indexed-cost}
For any integer table budget $T\ge1$, the tester uses
\begin{equation}\label{eq:indexed-cost}
Q=O(k\log(k+1)+kT\log(k+1)+kT)
\end{equation}
membership queries and distribution samples in total.
\end{thm}
\begin{proof}
At most $2k+1$ instances are created and each is checked at most once,
at cost $O(\log(k+1))$. Each passing instance builds one $T$-row
table, at $O(\log(k+1))$ queries per row. Reuse makes no additional
literal-location queries. At most $k+1$ searches each use at most
$T$ samples and $2T$ comparison queries; skipped rows use fewer.
The successful binary searches cost $O(k\log(k+1))$ queries in
total. These bounds hold on every input and every realization of
the randomness. Substituting \eqref{eq:short-budget} gives
$O(k\log^2(k+1)/\epsilon)$, whereas substituting
\eqref{eq:general-budget} gives $O(k^2\log(k+1)/\epsilon)$.
\end{proof}

The short-budget cost is near-linear, but its proved soundness scope
is hidden symmetry. The enlarged-budget cost and its proof cover
arbitrary Boolean functions. Shrinking $\epsilon$ alone does not
remove the extra factor of $k$ from the latter multiplicative bound.
The small-$\epsilon$ guarantee instead uses the additive bound of the
fresh-batch variant.

\section{A fresh-batch tester with an additive accuracy bound}
\label{batch:section}

This section gives a general-input variant of the relevant-block
framework. It uses fresh random bases in different batches and reuses
the resulting masks only within a batch. Joint literal localization
replaces separate orientation calls, and a global batch budget replaces
acceptance after a failed search stage. Thus the result below is not a
short-budget soundness claim for the original indexed-table reuse
algorithm. All distances are measured under the arbitrary, unknown
distribution $\mathcal D$, unless uniform distance is specified.
Both membership queries and samples from $\mathcal D$ are counted.

\begin{thm}\label{batch:main}
For $1\le k\le n$, $0<\epsilon\le1$, and $0<\delta<1$, there is an
adaptive distribution-free tester that always accepts every Boolean
$k$-junta and rejects every function $\epsilon$-far from all $k$-juntas
with probability at least $1-\delta$. Its query complexity is
\[
 O\!\left((k+\ln(1/\delta))
              \left(k\ln(k+1)+\frac1\epsilon\right)\right),
\]
independently of $n$. It draws
$O((k+\ln(1/\delta))/\epsilon)$ samples from $\mathcal D$.
In particular, for constant $\delta$ the total query complexity is
\[
 O\!\left(k^2\ln(k+1)+\frac{k}{\epsilon}\right).
\]
\end{thm}

\begin{cor}\label{batch:regimes}
For constant $\delta$, the tester in Theorem~\ref{batch:main} uses
$O((k/\epsilon)\ln(k+1))=\widetilde O(k/\epsilon)$ queries when
$\epsilon\le1/k$, and $O(k/\epsilon)$ queries when
$\epsilon\le1/(k\ln(k+1))$.
\end{cor}

\subsection{Certified block instances}
\label{batch:instances}

A current block instance consists of a nonempty block $B$ and an
observed full-flip distinguishing pair $a,a^{(B)}$:
$f(a)\ne f(a^{(B)})$. Its anchored restriction is fixed as
\[
 g_B(u)=f(u_Ba_{\bar B}),\qquad u\in\{0,1\}^B.
\]
Changing the anchor creates a new instance. Call an instance special
if $g_B$ is uniformly $1/8$-close to a nonconstant literal. This is an
analytic property; a declaration by the algorithm can be erroneous.
A state is safe when all current instances are special. The representative
literal of a special instance is unique: distinct literals have uniform
distance at least $1/2$, whereas two $1/8$ approximations would give
distance at most $1/4$ by the triangle inequality. Current blocks
are always disjoint. Every full-flip certificate guarantees a relevant
coordinate in its block, so $k+1$ such blocks certify rejection
deterministically.

We use the certified uniform tester in Theorem~2.3 of
\cite{liu2019distribution}. With junta parameter one and distance
parameter $1/8$, it uses $O(1)$ queries, always accepts a one-junta,
and rejects a function $1/8$-far from all one-juntas with probability
at least $2/3$. A rejection supplies two disjoint blocks with actual
distinguishing pairs. Shrinking each supplied block to the difference
set of its pair gives two nonempty disjoint full-flip certificates.
The children need not cover their parent; unused coordinates return
to the outside set.

\paragraph{\textsc{BatchCheck}$(B,a,\eta)$.}
Run the certified uniform tester independently
$r_0=\lceil\log_3(1/\eta)\rceil$ times on $g_B$.
On any rejection, lift the two shrunken distinguishing pairs to full
inputs using the fixed outside anchor, and return the certified split.
Otherwise perform $s_0=\lceil\log_4(1/\eta)\rceil$ independent
complement trials. In a trial choose a fair subset $P\subseteq B$,
put $Q=B\setminus P$, and query $u=a^{(P)}$ and $v=a^{(Q)}$.
If $f(u)=f(v)$, choose
$b\in\{a,a^{(B)}\}$ with the opposite value and return the pairs
$(b,u)$ and $(b,v)$. Their difference sets are $P,Q$ in some order.
If no split is found, return \textsc{Special}.

\begin{lem}\label{batch:check}
Every split returned by \textsc{BatchCheck} is valid deterministically.
If its input instance is nonspecial, its probability of returning
\textsc{Special} is at most $\eta$. Its query cost is
$O(\ln(1/\eta))$.
\end{lem}
\begin{proof}
All returned pairs have unequal observed values and disjoint
difference sets. In a complement trial, neither difference set can be
empty if the test returns a split.

Suppose the instance is nonspecial. If $g_B$ is also farther than
$1/8$ from both constants, it is $1/8$-far from every one-junta.
All uniform tester calls accept with probability at most
$3^{-r_0}\le\eta$.
Otherwise $g_B$ is at most $1/8$-far from a constant $c$.
The $B$-parts of $u,v$ are each uniform, although dependent, and hence
\[
 \Pr[f(u)=f(v)=c]\ge1-2/8=3/4.
\]
The existing certified endpoints have opposite values, so equal
values at $u,v$ yield the indicated split. Independent complement
trials all miss with probability at most $4^{-s_0}\le\eta$.
These are alternative cases; their error bounds are not added.
\end{proof}

\subsection{Joint literal localization}
\label{batch:localization}

\begin{lem}[Partial randomization]\label{batch:randomization}
If $f$ is $\epsilon$-far from all $k$-juntas under $\mathcal D$, then
for every fixed $I\subseteq[n]$ with $|I|\le k$, independent
$X\sim\mathcal D$ and a fair subset $R\subseteq\bar I$ satisfy
\[
 \Pr[f(X)\ne f(X^{(R)})]\ge\epsilon.
\]
\end{lem}
\begin{proof}
For each fixed $w$, the function $h_w(x)=f(x_Iw_{\bar I})$ is a
$k$-junta, so $\Pr_{X\sim\mathcal D}[f(X)\ne h_w(X)]\ge\epsilon$.
Average over uniform $w$. Conditional on $X$, $X_Iw_{\bar I}$ has
the same distribution as $X^{(R)}$.
\end{proof}

Identify $\{0,1\}$ with $\mathbb F_2$. For fixed subsets
$P_1,\ldots,P_r\subseteq B$ and $\alpha\in\mathbb F_2^r$, write
\[
 P_\alpha=\mathop{\triangle}_{\ell:\alpha_\ell=1}P_\ell,
 \qquad
 a_i=(\mathbf1\{i\in P_\ell\})_{\ell=1}^r.
\]
Here $\triangle$ denotes symmetric difference and an empty sum is
empty. The vector $a_i$ is the signature of coordinate $i$.

\paragraph{\textsc{DecodeSignature}$(g_B,P_1,\ldots,P_r,q)$.}
For $t=1,\ldots,q$, draw independent uniform
$\alpha_t\in\mathbb F_2^r$ and $U_t\in\{0,1\}^B$, and query
\[
 z_t=g_B(U_t)\oplus g_B(U_t^{(P_{\alpha_t})}).
\]
Return a maximizer $\widehat a\in\mathbb F_2^r$ of
$\sum_{t=1}^q\mathbf1\{z_t=\langle\alpha_t,\widehat a\rangle\}$,
using a fixed order to break ties. This requires $2q$ queries.

\begin{lem}[Joint decoding]\label{batch:decode}
If $g_B$ is $1/8$-close to a literal on $i$, then for every fixed
choice of $P_1,\ldots,P_r$,
\[
 \Pr[\widehat a\ne a_i\mid P_1,\ldots,P_r]
 \le(2^r-1)e^{-q/32}.
\]
\end{lem}
\begin{proof}
For every fixed $\alpha$, both $U$ and $U^{(P_\alpha)}$ are uniform.
A union bound on their two disagreements with the literal gives
\[
 \Pr\big[g_B(U)\oplus g_B(U^{(P_\alpha)})
            \ne\langle\alpha,a_i\rangle\big]\le1/4.
\]
The sign of a negated literal cancels in the XOR. Noise may depend on
$\alpha$; no independence of noise and $\alpha$ is assumed.
For a false candidate $a'\ne a_i$, put
\[
 X_t=\mathbf1\{z_t=\langle\alpha_t,a_i\rangle\}
     -\mathbf1\{z_t=\langle\alpha_t,a'\rangle\}.
\]
The two predictions differ on exactly half of all $\alpha_t$.
On each such $\alpha_t$, the conditional expectation of $X_t$ is at
least $1/2$; otherwise $X_t=0$. Thus $\mathbb E X_t\ge1/4$.
Conditional on the fixed subsets, the $X_t\in[-1,1]$ are independent.
Hoeffding's inequality gives
$\Pr[\sum_tX_t\le0]\le e^{-q/32}$.
A false decoded candidate must tie or beat the true score; union
bounding over the $2^r-1$ false candidates proves the claim.
\end{proof}

\subsection{A fresh batch}
\label{batch:algorithm}

Set
\begin{equation}\label{batch:parameters}
 L=\left\lceil\frac{16}{\epsilon}\right\rceil,
 \qquad r=\lceil\log_2(L+1)\rceil,
 \qquad q=\left\lceil32\big(r\ln2+\ln(16k)\big)\right\rceil.
\end{equation}
Fix $L$ distinct nonzero vectors
$c_1,\ldots,c_L\in\mathbb F_2^r$, possible because $2^r-1\ge L$.

\paragraph{\textsc{FreshBatch}.}
The input is a fixed collection of $m\le k$ certified block instances
$B_1,\ldots,B_m$ declared special.
\begin{enumerate}
\item Independently for each $j$, generate $r$ independent fair
subsets $P_{j,1},\ldots,P_{j,r}$ of $B_j$ and use
\textsc{DecodeSignature} with parameter $q$ to obtain
$\widehat a_j$. Define
\[
 M_{j,\ell}=
 \begin{cases}
 P_{j,\ell},&\widehat a_{j,\ell}=0,\\
 B_j\setminus P_{j,\ell},&\widehat a_{j,\ell}=1.
 \end{cases}
\]
\item For $s=1,\ldots,L$, form
$N_{j,s}=\mathop{\triangle}_{\ell:c_{s,\ell}=1}M_{j,\ell}$.
Draw fresh $x_s\sim\mathcal D$ and a fair subset $C_s$ of
$[n]\setminus\bigcup_j B_j$. Put
$R_s=C_s\cup\bigcup_jN_{j,s}$ and query
$f(x_s),f(x_s^{(R_s)})$.
\item At the first discrepancy, perform block binary search on the
nonempty chunks among $C_s,N_{1,s},\ldots,N_{m,s}$ to find a chunk
$H$ with a full-flip certificate. If $H=C_s$, return addition of
this new certified block.
If $H=N_{h,s}$, put $K=B_h\setminus N_{h,s}$ and try twice, with
new independent uniform $u\in\{0,1\}^{B_h}$, to find
$g_{B_h}(u)\ne g_{B_h}(u^{(K)})$.
On a discrepancy, lift this pair using the old anchor and return
replacement of $B_h$ by the certified children $N_{h,s},K$.
If both attempts fail, return \textsc{Fail} without changing the
structural state; do not continue to another discrepancy.
\item If no probe produces a discrepancy, return \textsc{Fail}.
\end{enumerate}
Discard all bases, localization data, and masks at batch end. New
instances receive the anchors of their actual certificates and are
initially unchecked. All randomness of the next batch is fresh.

Block binary search uses at most $\lceil\log_2(k+1)\rceil$ further
queries. It splits the remaining chunks into two groups, flips the
union of the first group at one endpoint, and keeps the half whose
endpoints disagree. Its maintained pair differs exactly on the union
of the remaining chunks, so a leaf gives a full-flip certificate for
one chunk.

\begin{lem}[Ideal masks]\label{batch:pairwise}
Fix a safe state, let $i_j$ be the literal coordinate of $g_{B_j}$,
and put $I=\{i_1,\ldots,i_m\}$. Replace each decoded signature by its
true signature, keeping the same bases and coefficients. The resulting
ideal masks $R_1^*,\ldots,R_L^*$ are pairwise independent fair subsets
of $\bar I$.
\end{lem}
\begin{proof}
For one fair subset $P\subseteq B_j$, the side of
$P,B_j\setminus P$ excluding $i_j$ has independent fair membership
bits on $B_j\setminus\{i_j\}$. This follows by conditioning on the
membership bit of $i_j$: the opposite-membership bits of all other
coordinates remain independent and fair, and their joint law does not
depend on that conditioning bit. Bases and blocks are independent.
Thus the membership vector $b_v\in\mathbb F_2^r$ of each
$v\in B_j\setminus\{i_j\}$ in its ideal base masks is uniform,
independently across all such coordinates. Its membership in row $s$
is $\langle b_v,c_s\rangle$. This bit is fair since $c_s\ne0$.
Distinct nonzero $c_s,c_t$ are linearly independent over $\mathbb F_2$,
so their two membership bits are uniform on $\mathbb F_2^2$.
Independence across coordinates gives independence of the two entire
masks. The outside masks $C_s$ are independent and fresh. No rank
condition on the realized family of base subsets is needed.
\end{proof}

\begin{lem}[Certified progress]\label{batch:progress}
On an $\epsilon$-far input, a fresh batch started in any fixed safe
state with at most $k$ blocks returns certified progress with
probability at least $13/16$.
\end{lem}
\begin{proof}
Pregenerate all $L$ samples and outside masks for the analysis.
In the ideal experiment let
$Y_s=\mathbf1\{f(x_s)\ne f(x_s^{(R_s^*)})\}$ and $Z=\sum_sY_s$.
Lemmas~\ref{batch:randomization} and~\ref{batch:pairwise} give
pairwise independent Bernoulli variables of common mean
$p\ge\epsilon$. Therefore
\begin{equation}\label{batch:nohit}
 \Pr[Z=0]
 \le\frac{\operatorname{Var}(Z)}{(\mathbb E Z)^2}
 =\frac{1-p}{Lp}\le\frac1{L\epsilon}\le\frac1{16}.
\end{equation}
Conditional on every fixed basis, Lemma~\ref{batch:decode} bounds
each decoding error by $1/(16k)$. A union bound over at most $k$
blocks gives total decoding error at most $1/16$. If no decoding error
occurs, all actual masks equal the ideal masks. An actual discrepancy
is therefore reached except on events of total probability at most
$1/8$. This is an unconditional coupling; it does not condition on
correct decoding and then assume the masks remain pairwise independent.

An outside hit gives immediate progress. For an inside hit, correct
decoding gives $i_h\in K$. For every such fixed $K$, the literal flips
between uniform $u$ and $u^{(K)}$, and hence
\[
 \Pr[g_{B_h}(u)\ne g_{B_h}(u^{(K)})]\ge1-2/8=3/4.
\]
The certificate samples are fresh after the adaptive choice of $h,K$;
conditional on that selection history the two tests both miss with
probability at most $1/16$. Adding the three bad-event bounds gives
failure probability at most $3/16$.
\end{proof}

Even with an incorrect decoder or an unsafe state, an update is valid:
binary search and the complement tests use actual unequal function
values. Each returned update increases the number of disjoint
certified blocks by one. In particular, an empty complement cannot
pass the certificate test. The decoder itself supplies no literal-side
certificate; the fresh complement tests are necessary for this reason.

\subsection{Global budget and soundness}
\label{batch:global}

Set
\begin{equation}\label{batch:outerparameters}
 \eta=\frac{\delta}{4(k+1)},\qquad
 M=\left\lceil\frac{32}{13}
          \left(k\ln2+\ln\frac2\delta\right)\right\rceil.
\end{equation}
Initialize the block collection empty and a batch counter $t=0$.
Repeat:
\begin{enumerate}
\item Reject if there are $k+1$ certified blocks.
\item If an unchecked instance exists, run \textsc{BatchCheck}
with parameter $\eta$. On \textsc{Special}, mark it checked.
On a split, remove the parent and install two unchecked certified
children, returning unused coordinates to the outside set.
Return to step 1.
\item If all instances are checked and $t=M$, accept. Otherwise
increment $t$ and run a completely fresh \textsc{FreshBatch}.
Apply a returned update and mark new instances unchecked, or retain
the structural state on \textsc{Fail}. Return to step 1.
\end{enumerate}
A failed batch does not cause acceptance. The budget is global and is
not reset after a structural update.

\begin{proof}[Correctness in Theorem~\ref{batch:main}]
All blocks remain disjoint, nonempty, and certified. Each update
increases the count by one. A $k$-junta cannot yield $k+1$ such blocks,
so it is always accepted. There are at most $k+1$ updates and $M$
batches, each with a finite budget, giving termination.

For soundness, assume that $f$ is $\epsilon$-far from all $k$-juntas.
At most $2(k+1)$ instances are created, and each is checked once with
fresh randomness after its block and anchor have been fixed. Let
$\mathsf{Bad}$ denote any nonspecial instance being marked special.
By Lemma~\ref{batch:check} and an adaptive union bound,
\[
 \Pr[\mathsf{Bad}]\le2(k+1)\eta\le\delta/2.
\]

For analysis, virtually extend the execution to $M$ batches. Complete
all pending structural checks before each real batch and include them
in its preceding history. Define $V_t$ as the progress indicator for
batch $t$ if no structural error has occurred and execution is still
running. After a structural error or rejection, set all subsequent
virtual indicators to one. This convention is only for analysis;
the algorithm need not identify a structural error. For every history
before a real or virtual batch,
\[
 \Pr[V_t=1\mid\text{history}]\ge13/16,
 \qquad
 \mathbb E[2^{-V_t}\mid\text{history}]\le19/32.
\]
Indeed, a real batch before any error starts in a safe state and uses
fresh randomness, so Lemma~\ref{batch:progress} applies; otherwise
$V_t=1$ by definition. Iterated conditional expectation and Markov's
inequality yield
\begin{align*}
 \Pr\left[\sum_{t=1}^M V_t\le k\right]
 &\le2^k\mathbb E\left[2^{-\sum_{t=1}^M V_t}\right]\\
 &\le2^k(19/32)^M
 \le\exp(k\ln2-13M/32)\le\delta/2.
\end{align*}
If execution accepts and $\mathsf{Bad}$ does not occur, all $M$
batches were executed and their progress count is at most $k$.
Structural splits can only increase the block count further. Thus
\[
 \Pr[\text{accept}]
 \le\Pr[\mathsf{Bad}]+
    \Pr\left[\sum_{t=1}^M V_t\le k\right]\le\delta.
\]
We have not conditioned fresh randomness on the future event
$\mathsf{Bad}^c$, nor asserted independence of adaptively chosen states.
Only the conditional progress bound is used.
\end{proof}

\subsection{Query bound}
\label{batch:complexity}

The structural checks cost $O(k\ln(4(k+1)/\delta))$ queries in total.
A batch with $m\le k$ blocks uses at most
\[
 2mq+3L+\lceil\log_2(k+1)\rceil+4
\]
queries: joint decoding, at most $L$ samples and $2L$ evaluations,
one block binary search, and two complement trials. Consequently
the deterministic worst-case bound is
\begin{equation}\label{batch:cost}
 Q\le O\!\left(k\ln\frac{4(k+1)}\delta\right)
 +M\left(2kq+3L+\lceil\log_2(k+1)\rceil+4\right).
\end{equation}
Here $L=O(1/\epsilon)$, $r=O(\ln(2/\epsilon))$,
$q=O(\ln(2k/\epsilon))$, and $M=O(k+\ln(1/\delta))$.
The extra logarithm in $1/\epsilon$ is absorbed by the sampling term:
for $x=1/\epsilon$, the inequality $\ln(x/k)\le x/k-1$ gives
\[
 k\ln(2k/\epsilon)
 =k\ln(2k^2)+k\ln(x/k)
 \le k\ln(2k^2)+x-k
 =O(k\ln(k+1)+1/\epsilon).
\]
Also $k\ln(4(k+1)/\delta)=O(k\ln(k+1)+k\ln(1/\delta))$
is absorbed by the bound in Theorem~\ref{batch:main}.
Substitution in~\eqref{batch:cost} proves that theorem's query bound;
at most $ML$ samples are drawn. The two accuracy regimes in
Corollary~\ref{batch:regimes} follow by comparing the additive terms
$k^2\ln(k+1)$ and $k/\epsilon$.

This is a query bound, not an $n$-independent running-time bound.
Direct decoding enumerates $2^r=O(1/\epsilon)$ candidates and scores
them using $O(qr)$ elementary operations each. Constructing length-$n$
query strings also takes time depending on $n$. No matching lower
bound, global optimality claim, or general near-linear bound for
constant $\epsilon$ is asserted.

\section{Conclusions}
The indexed orientation tables have a complete near-linear analysis
for functions symmetric on an unknown set of variables. Under that
promise, a safe passed block contains one globally relevant variable,
and its cached nonliteral sides are irrelevant everywhere. For general
Boolean functions, stable deletion proves the indexed reuse framework
sound at an enlarged budget. A separately specified fresh-batch variant
provides $O(k^2\log(k+1)+k/\epsilon)$ queries and consequently
near-linear guarantees in the stated accuracy regimes.

The results preserve one-sided acceptance through actual disjoint-block
certificates. They do not prove a general near-linear one-sided bound
at constant $\epsilon$, nor do they assert a matching lower bound for
the restricted regimes. Establishing the shorter indexed-table budget
for more general function classes requires further analysis.

\clearpage
\appendix
\section{IsLiteral Algorithm and Its Guarantee}\label{sec:literal-check}
\begin{algorithm}[H]
\caption{$\IL(B,a)$ with $f(a)\ne f(a^{(B)})$}
\label{Framework:Algorithm for IsLiteral}

\begin{enumerate}[leftmargin=*,itemsep=3pt]
\item Freeze $g_B(u)=f(u_Ba_{\bar B})$, set $\gamma=1/8$, and
      repeat the following steps $q_0=\lceil\ln k+6\rceil$ times
      with independent fresh randomness.
\item Run $\UJ(g_B,1,\gamma)$. On rejection, take its two
      disjoint distinguishing pairs, shrink to their difference sets,
      lift with $a_{\bar B}$, and return the two certified children.
\item Draw a fair partition $B=P\sqcup Q$. Query
      $u=a^{(P)}$ and $v=a^{(Q)}$.
      If $f(u)=f(v)$, select $b\in\{a,a^{(B)}\}$ with the opposite
      value and return $\diff(b,u),\diff(b,v)$ with their certificates.
\item If no repetition splits, return pass.
\end{enumerate}
\end{algorithm}

\begin{proof}[Proof of Lemma~\ref{Literal_Lem}]
Uniform-test splits are valid by the shrinking construction.
In the random-partition branch, $b$ exists since the original endpoints
have opposite values. The children are $P,Q$ in some order and cannot
be empty because their queried labels differ from $f(b)$.

Suppose $g_B$ is farther than $\gamma$ from every literal. If it is
also farther than $\gamma$ from both constants, it is
$\gamma$-far from every $1$-junta, and one uniform-test call rejects
with probability at least $2/3$. Otherwise $g_B$ is $\gamma$-close
to a constant $c$. The random endpoints $u_B,v_B$ each have a uniform
marginal, although they need not be independent. Hence
\[
\Pr[f(u)=f(v)=c]\ge1-2\gamma=3/4.
\]
This event triggers a certified split. In either case a repetition
misses with probability at most $1/3$. Independence across repetitions
therefore gives
\[
\Pr(\mathsf{pass})\le3^{-q_0}
\le3^{-6}k^{-\ln3}\le\frac1{28k}.
\]
The query cost is $O(q_0)$ at fixed $\gamma$. For an exact literal,
the uniform tester always accepts and the two random-partition
endpoints have opposite values, so the procedure always passes.
\end{proof}


\begin{thebibliography}{99}

\bibitem{goldreich1998property}
Oded Goldreich, Shafi Goldwasser, and Dana Ron.
\newblock Property testing and its connection to learning and approximation.
\newblock \emph{Journal of the ACM}, 45(4):653--750, 1998.
\newblock \url{https://doi.org/10.1145/285055.285060}.

\bibitem{halevy2007distribution}
Shirley Halevy and Eyal Kushilevitz.
\newblock Distribution-free property-testing.
\newblock \emph{SIAM Journal on Computing}, 37(4):1107--1138, 2007.
\newblock \url{https://doi.org/10.1137/050645804}.

\bibitem{parnas2002testing}
Michal Parnas, Dana Ron, and Alex Samorodnitsky.
\newblock Testing basic Boolean formulae.
\newblock \emph{SIAM Journal on Discrete Mathematics}, 16(1):20--46, 2002.
\newblock \url{https://doi.org/10.1137/S0895480101407444}.

\bibitem{fischer2004testing}
Eldar Fischer, Guy Kindler, Dana Ron, Shmuel Safra, and Alex Samorodnitsky.
\newblock Testing juntas.
\newblock \emph{Journal of Computer and System Sciences}, 68(4):753--787, 2004.
\newblock \url{https://doi.org/10.1016/j.jcss.2003.11.004}.

\bibitem{diakonikolas2007concise}
Ilias Diakonikolas, Homin K. Lee, Kevin Matulef, Krzysztof Onak,
Ronitt Rubinfeld, Rocco A. Servedio, and Andrew Wan.
\newblock Testing for concise representations.
\newblock In \emph{Proceedings of the 48th Annual IEEE Symposium on
Foundations of Computer Science (FOCS)}, pages 549--558, 2007.
\newblock \url{https://doi.org/10.1109/FOCS.2007.32}.

\bibitem{blais2008improved}
Eric Blais.
\newblock Improved bounds for testing juntas.
\newblock In \emph{Approximation, Randomization and Combinatorial
Optimization. Algorithms and Techniques (APPROX/RANDOM 2008)},
volume 5171 of \emph{Lecture Notes in Computer Science},
pages 317--330. Springer, 2008.
\newblock \url{https://doi.org/10.1007/978-3-540-85363-3_26}.

\bibitem{blais2009testing}
Eric Blais.
\newblock Testing juntas nearly optimally.
\newblock In \emph{Proceedings of the 41st Annual ACM Symposium on Theory
of Computing (STOC)}, pages 151--158, 2009.
\newblock \url{https://doi.org/10.1145/1536414.1536437}.

\bibitem{sauglam2018near}
Mert Sa\u{g}lam.
\newblock Near log-convexity of measured heat in (discrete) time and
consequences.
\newblock In \emph{2018 IEEE 59th Annual Symposium on Foundations of
Computer Science (FOCS)}, pages 967--978, 2018.
\newblock \url{https://arxiv.org/abs/1808.06717}.

\bibitem{liu2019distribution}
Zhengyang Liu, Xi Chen, Rocco A. Servedio, Ying Sheng, and Jinyu Xie.
\newblock Distribution-free junta testing.
\newblock \emph{ACM Transactions on Algorithms}, 15(1), Article 1, pages 1:1--1:23, 2019.
\newblock \url{https://doi.org/10.1145/3264434}.
\newblock The certificate-producing uniform tester used here is
Theorem 2.3 of the arXiv full version,
\url{https://arxiv.org/abs/1802.04859}.

\bibitem{bshouty2019almost}
Nader H. Bshouty.
\newblock Almost optimal distribution-free junta testing.
\newblock In \emph{34th Computational Complexity Conference (CCC 2019)},
LIPIcs volume 137, pages 2:1--2:13, 2019.
\newblock \url{https://doi.org/10.4230/LIPIcs.CCC.2019.2}.

\bibitem{bshouty2020concise}
Nader H. Bshouty.
\newblock Almost optimal testers for concise representations.
\newblock In \emph{Approximation, Randomization, and Combinatorial
Optimization. Algorithms and Techniques (APPROX/RANDOM 2020)},
volume 176 of \emph{Leibniz International Proceedings in Informatics},
pages 5:1--5:20, 2020.
\newblock \url{https://doi.org/10.4230/LIPIcs.APPROX/RANDOM.2020.5}.

\bibitem{chen2025relative}
Xi Chen, William Pires, Toniann Pitassi, and Rocco A. Servedio.
\newblock Testing juntas and junta subclasses with relative error.
\newblock In \emph{Proceedings of the 38th Conference on Learning Theory},
volume 291 of \emph{Proceedings of Machine Learning Research},
pages 5214--5245, 2025.
\newblock \url{https://proceedings.mlr.press/v291/chen25h.html}.

\bibitem{beretta2025feature}
Lorenzo Beretta, Nathaniel Harms, and Caleb Koch.
\newblock Feature selection and junta testing are statistically equivalent.
\newblock arXiv preprint arXiv:2505.04604, 2025.
\newblock \url{https://arxiv.org/abs/2505.04604}.

\bibitem{bshouty2026small}
Nader H. Bshouty and George Haddad.
\newblock Classes testable with $O(1/\epsilon)$ queries for small
$\epsilon$ independent of the number of variables.
\newblock arXiv preprint arXiv:2604.05615, 2026.
\newblock \url{https://arxiv.org/abs/2604.05615}.

\bibitem{odonnell2014analysis}
Ryan O'Donnell.
\newblock \emph{Analysis of Boolean Functions}.
\newblock Cambridge University Press, 2014.
\newblock \url{https://doi.org/10.1017/CBO9781139814782}.
\newblock An author-provided edition is available as
\url{https://arxiv.org/abs/2105.10386}.

\bibitem{hoeffding1963probability}
Wassily Hoeffding.
\newblock Probability inequalities for sums of bounded random variables.
\newblock \emph{Journal of the American Statistical Association},
58(301):13--30, 1963.
\newblock \url{https://doi.org/10.1080/01621459.1963.10500830}.

\end{thebibliography}
\end{document}